\documentclass[12pt, reqno]{amsart}
\usepackage{amsmath, amsthm, amscd, amsfonts, amssymb, graphicx, color}
\usepackage[bookmarksnumbered, colorlinks, plainpages]{hyperref}
\hypersetup{colorlinks=true,linkcolor=red, anchorcolor=green, citecolor=cyan, urlcolor=red, filecolor=magenta, pdftoolbar=true}

\newtheorem{theorem}{Theorem}[section]
\newtheorem{lemma}[theorem]{Lemma}
\newtheorem{proposition}[theorem]{Proposition}

\theoremstyle{definition}
\newtheorem{definition}[theorem]{Definition}

\theoremstyle{remark}

\numberwithin{equation}{section}

\begin{document}

\setcounter{page}{1}

\title[Uniform estimates for oscillatory integrals ...]{Uniform estimates for oscillatory integrals with homogeneous polynomial phases of degree 4}

\author[M. Ruzhansky, A.R. Safarov, G.A. Khasanov]{Michael Ruzhansky, Akbar R. Safarov$^{*}$, Gafurjan A. Khasanov}

\address{\textcolor[rgb]{0.00,0.00,0.84}{Michael Ruzhansky\newline Department of Mathematics: Analysis, Logic and Discrete Mathematics  \\
  \newline  Ghent University  \newline  Krijgslaan 281, Ghent, Belgium  \newline School of Mathematical Sciences, Queen Mary University of London  \newline United Kingdom}}
\email{\textcolor[rgb]{0.00,0.00,0.84}{michael.ruzhansky@ugent.be}}
\address{\textcolor[rgb]{0.00,0.00,0.84}{Akbar R.Safarov \newline Institute of Mathematics named after V. I. Romanovskiy at the Academy of Sciences of the Republic of Uzbekistan    \newline Olmazor district, University 46, Tashkent, Uzbekistan \newline Samarkand State University \newline Department Mathematics,
15 University Boulevard
 \newline  Samarkand, 140104, Uzbekistan }}
\email{\textcolor[rgb]{0.00,0.00,0.84}{safarov-akbar@mail.ru}}

\address{\textcolor[rgb]{0.00,0.00,0.84}{Gafurjan A.Khasanov \newline  Samarkand State University \newline Department Mathematics,
15 University Boulevard
 \newline  Samarkand, 140104, Uzbekistan }}
\email{\textcolor[rgb]{0.00,0.00,0.84}{khasanov-g75@mail.ru}}

\thanks{All authors contributed equally to the writing of this paper. All authors read and approved the final manuscript.}

\let\thefootnote\relax\footnote{$^{*}$Corresponding author}



\begin{abstract}In this paper  we consider the uniform estimates for oscillatory integrals with a two-order homogeneous polynomial phase.  The estimate is sharp and the result is an analogue of the more general theorem of V. N. Karpushkin \cite{Karpushkin1983} for sufficiently smooth functions.
\end{abstract}
\maketitle
\tableofcontents
\section{Introduction}

In this paper we continue the study of oscillatory integrals with smooth phase functions initiated in \cite{SafUfim}, in the case when the phase function is a homogeneous polynomial of degree four in two variables.

\begin{definition}
An oscillatory integral with phase $f$ and amplitude $a$ is an integral of the form
\begin{equation}\label{int1}
J(\lambda ,f,a)=\int _{{\bf {\mathbb R}}^{n} }a(x)e^{i\lambda f(x)} dx,
\end{equation}
where $a\in C_{0}^{\infty}({\bf {\mathbb R}}^{n})$ and $\lambda\in{\bf{\mathbb R}}$.
\end{definition}

If the support of $a$ lies in a sufficiently small neighborhood of the origin and $f$ is an analytic function at $x=0,$ then for $\lambda\rightarrow\infty$
the following  asymptotic expansion  holds (\cite{Karpushkin1980}):
 \begin{equation}\label{int22}
J(\lambda,f,a)\approx e^{i\lambda f(0)}\sum_{s}\sum_{k=0}^{n-1}b_{s,k}(a)\lambda^{s}(\ln\lambda)^{k},
\end{equation}
where $s$ belongs to a finite number of arithmetic progressions, independent of
$a,$ composed of negative rational numbers.

\begin{definition}
The oscillation exponent $f$ at the point $0$ is the number $\beta(f),$ maximum among all numbers $s$ possessing the following property: for any neighborhood of the point $0$ there exists a function $a$ supported in this neighborhood for which in decomposition \eqref{int22} there is $k$ such that $b_{s,k}(a)\neq0.$
\end{definition}

\begin{definition}
The multiplicity $p(f)$ of the oscillation exponent $\beta(f)$ is the maximal $p$  with the property: for any neighborhood of the point $0$ there exists $a$ with support in it, for which $b_{\beta(f),p(f)}(a)\neq0.$ The pair $(\beta(f),p(f))$ is denoted by $O(f).$
\end{definition}

 Let $U\subset V\subset \mathbb R^{2}$ be bounded neighborhoods of the origin, $\overline{U}(\overline{V})$ the closure of $U(V)$, respectfully. Suppose that the function $f:\overline{V}\to \mathbb R$ (where $f\in C^{N} \left(\overline{V}\right),$ $\left(N\ge 8\right)$) has the form:
\begin{equation} \label{GrindEQ__1_}
f\left(x_{1},x_{2}\right)=f_{\pi}\left(x_{1},x_{2}\right)+g\left(x_{1},x_{2}\right),
\end{equation}
where $f_{\pi}\left(x_{1},x_{2}\right)$ is a homogeneous polynomial of degree 4 having the root $b(0,0)$ of multiplicity at most 2 (e.g. polynomial $f_{\pi}(x)$ of order 4  has at most two roots of multiplicity two on the unit circle in $\mathbb R^{2}$ centered at the origin), and $g\in C^{N} \left(V\right)$ is such that $D^{\alpha} g(0,0)=0$ for all $\alpha_{1} +\alpha_{2}\le 4,$ where $D^{\alpha }$ is $D^{\alpha } =\frac{\partial^{\left|\alpha\right|}}{\partial x_{1}^{\alpha_{1}} \partial x_{2}^{\alpha_{2}}}$, $\alpha:=(\alpha_{1},\alpha_{2})\in \mathbb{Z}_{+}^{2}$ is a multi-index, $\mathbb{Z}_{+}=\{0\}\cup\mathbb{N}$ are the non-negative integers.

\begin{definition} Let $F\in C^{N} \left(\overline{V}\right)$ be a function such that $\left\| F\right\| _{C^{N} \left(\overline{V}\right)} <\varepsilon $, where $N$ is a natural number and $\varepsilon$ is  a sufficiently small positive number, where $\left\| F\right\| _{C^{N} \left(\overline{V}\right)} =\mathop{\max }\limits_{\overline{V}} \sum\limits_{\left|\alpha \right|\le N}\left|\frac{\partial^{\left|\alpha \right|} F\left(x_{1}, x_{2} \right)}{\partial x_{1}^{\alpha _{1} } \partial x_{2}^{\alpha _{2} } } \right|.$ Then the function $f+F$ is called to be a deformation of $f$ (see \cite{Karpushkin1983}).
\end{definition}

\begin{definition}(\cite{Karpushkin1983})
Let $f:V\longrightarrow \mathbb{R}$ be a $C^{N}(\overline{V})$ function. We say that the oscillatory integral with phase $f$ has the uniform estimate
$(\beta_{u},p_{u}),$ where $\beta_{u}\leq0,$ $p_{u}\geq0,$\,\,\
$p_{u}$ is an integer, if there are $C_{V}>0,$ $\epsilon_{V}>0$ and a neighborhood
$U$ of $0\in \mathbb{R}^{2},$  such that $U\subset\subset V$ and
for all $\lambda\geq2,$ $\|F\|_{C^{N}(U)}<\epsilon,$ $a\in
C_{0}^{2}(U)$, we have
$$|J(\lambda,f+F,a)|\leq C_{V}\lambda^{\beta_{u}}(\ln\lambda)^{p_{u}}\|a\|_{C^{2}}.$$
\end{definition}

 The main result of the work is the following.

\begin{theorem}\label{theorem1} Let $f\in C^{8} (\bar{V})$ have the form \eqref{GrindEQ__1_}.
Then there exists a positive number $\varepsilon$ and a neighborhood $U\subset V$ of the origin such that for any functions $a\in C_{0}^{1} \left(U\right)$ and $F\in C^{8}\left(\overline{V}\right),$   $\left\|F\right\|_{C^{8}\left(\overline{V}\right)}<\varepsilon,$ the following estimate  holds:
\begin{equation}\label{Form1}
\left|\int _{U}e^{i\lambda \left(f+F\right)}a\left(x\right)dx\right|\le\frac{C\left\|a\right\|_{C^{1}}\ln(2+|\lambda|)}{\left|\lambda \right|^{\frac{1}{2}}}.
\end{equation}
\end{theorem}
1) In paper V. N. Karpushkin \cite{Karpushkin1984} showed that,  if $f$ is analytic at the origin, $df|_{0}=0$  and $d^{2}f|_{0}$ has corank 2, then the result of this theorem holds. Here, we show that an analog of V. N. Karpushkin's result \cite{Karpushkin1984}  (and also see \cite{Varch})  holds true for sufficiently smooth phase function.

2)  Note that if $g\equiv0,$  $F\equiv0,$ and $a(0)\neq0,$ then we have \cite{IkromovMuller2011}
$$\lim\limits_{\lambda\rightarrow+\infty}\frac{\lambda^{\frac{1}{2}}}{\ln\lambda}\int_{\mathbb{R}^{2}}e^{if(x)}a(x)dx=ca(0)$$ with $c\neq0.$

\section{Auxiliary statements}

We first present some simple auxiliary definitions.
\begin{definition}(\cite{Arn})
Consider the arithmetic space $\mathbb{R}^{n}$ with fixed coordinates $x_{1}, x_{2},\dots, x_{n}.$ A function $f:\mathbb{R}^{n}\setminus \{0\}\rightarrow \mathbb{R}$ is called to be a quasi-homogeneous function of degree $d$ with exponents $\alpha_{1}, \alpha_{2},$ $\dots,\alpha_{n},$ if for any $\lambda>0$ and $x\neq0$ we have $f(\lambda^{\alpha_{1}} x_{1}, \lambda^{\alpha_{2}} x_{2}, \dots, \lambda^{\alpha_{n}} x_{n})= \lambda^{d} f(x_{1}, x_{2},\dots, x_{n}).$ The powers $\alpha_{j}$ are called weights of the variables $x_{j}.$
\end{definition}

\begin{definition}
The standard gradient ideal for a smooth function $f:\mathbb{R}^{n}\rightarrow \mathbb{R}$, such that $f(0)=0$ and  $I_{\nabla f} =\left\langle\frac {\partial f}{\partial x_{1}}, \frac{\partial f} {\partial x_{2}},\dots, \frac {\partial f} {\partial x_{n}} \right\rangle$ is the gradient ideal, an ideal generated by $\frac {\partial f}{\partial x_{1}}, \frac{\partial f} {\partial x_{2}},\dots, \frac {\partial f} {\partial x_{n}}.$
\end{definition}

 Let the function $f\left(x_{1}, x_{2}\right)$ satisfy the following conditions:

1)$\frac{\partial^{\left|\alpha \right|} f\left(0,0\right)}{\partial x_{1}^{\alpha _{1}}\partial x_{2}^{\alpha_{2}}}=0,$   for all  $\left|\alpha\right|=\alpha_{1}+\alpha_{2}\le 3.$
2)  the equation $f_{4}\left(x_{1},x_{2}\right)=0$ on $\mathbb{S}_{1}$ (where $f_{4}$ is the Taylor polynomial of the function $f$ of order 4, and $\mathbb{S}_{1}$ is the unit circle in $\mathbb R^{2}$ centered at the origin) has at most two roots.

Then the function $f$, by a linear transformation, can be reduced to the form $f\left(x\left(u\right)\right)=u_{1}^{4}+\mu u_{1}^{2}u_{2}^{2}+u_{2}^{4}+g\left(u_{1} ,u_{2} \right)$ or $f\left(x\left(u\right)\right)=u_{1}^{2}(u_{1}^{2}\pm u_{2}^{2}),$ where $g\left(u_{1} ,u_{2} \right)$  is a  function satisfying the condition $D^{\alpha}g\left(0,0\right)=0,$ for all $\left|\alpha \right|\le 4$.
The main part with the weight $(\frac{1}{4},\frac{1}{4})$  of the function $f$ is denoted by $f_{\pi}$. If $\mu^{2}\neq4$ then $f_{\pi}$ has $X_{9}$ type singularity \cite{Arn}. In this case we have  the estimate \eqref{Form1} without "log" multiplier for the case when $f$ and $F$ are smooth function.

 Following \cite{Karpushkin1983}, we denote by $E_{d}$ the linear space of polynomials of degree less than $d$ with respect to the weight $\left(\frac{1}{4},\frac{1}{4} \right),$ and let \textit{$I_{\nabla f_{\pi}}$} be the standard gradient ideal of the function $f_{\pi}$.

\begin{definition}
Following \cite{Karpushkin1983} the coordinate subspace $B\subset E_{1}$ is said to be a versal subspace if $\left(I_{\nabla f_{\pi } } \bigcap E_{1}\right)\oplus B=E_{1},$ that is, $I_{\nabla f_{\pi } } \bigcap E_{1} \bigcap B=\{0\}$ and $\left(I_{\nabla f_{\pi}}\bigcap E_{1} \right)+B=E_{1}$).
\end{definition}

It is easy to show that $B=$Span$\left\langle 1,x_{1} ,x_{2} ,x_{1}^{2},x_{2}^{2},x_{1}x_{2},x_{1}^{3},x_{2}^{3} \right\rangle$ is a versal subspace $B$ for the function $f_{\pi}\left(x_{1},x_{2}\right)$ (see Lemma \ref{Form2}).

Let $\pi_{d}(F)=\sum\limits_{\frac{m_{1}}{4}+\frac{m_{2}}{4}<d}s_{m} x^{m}$  be the Taylor polynomial at the point 0 of the function $F.$ Thus, $\pi_{1}$ defines a mapping of the space $C^{N}\left(V\right)$ onto the space $E_{d},$ where $d\le \frac{N}{4}.$

\begin{lemma}\label{Form2}
Let $f_{\pi}$ be given in  one of the  following forms: $x_{1}^{4}+\mu x_{1}^{2}x_{2}^{2}+x_{2}^{4}$ or $ x_{1}^{2}(x_{1}^{2}\pm x_{2}^{2})$.
Then $B=$Span$\left\langle 1,x_{1},x_{2} ,x_{1}^{2},x_{2}^{2},x_{1}x_{2},x_{1}^{3},x_{2}^{3} \right\rangle$ is a versal subspace for $f_{\pi}.$
\end{lemma}
\begin{proof}
We have that $\partial_{1}f_{\pi}$ and $\partial_{2}f_{\pi}$ are linearly independent (see \cite{Karpushkin1983} Lemma 20, page 33) and $x_{1}^{3}$ and $x_{2}^{3}$ do not belong to the space spaned by  $\partial_{1}f_{\pi}$ and $\partial_{2}f_{\pi}$. Moreover, in the considered cases $\partial_{1}f_{\pi}$ and $\partial_{2}f_{\pi}$ are linearly independent and spaces spanned by  $\langle\partial_{1}f_{\pi},\partial_{2}f_{\pi}\rangle$ and $\langle x_{1}^{3},x_{2}^{3}\rangle$ have only trivial intersection, that is $\langle\partial_{1}f_{\pi},\partial_{2}f_{\pi}\rangle\cap\langle x_{1}^{3},x_{2}^{3}\rangle=\{0\}.$

Consequently $B$ is a versal subspace of $E_{1}$ corresponding to $f_{\pi}.$ Lemma \ref{Form2} is proved.
\end{proof}

Let we denote $\vartheta_{C^{8}\left(\overline{V}\right)}(\varepsilon):=\{F\in C^{8}\left(\overline{V}\right),$  $\left\|F\right\|_{C^{8} \left(\overline{V}\right)}<\varepsilon\}.$ The next proposition on the possibility of smoothly choosing a change of coordinates is an analogue of the versality theorem. An analogue of the deformation of the function of $f$ is $f+F$, where $F\in\vartheta_{C^{8}\left(\overline{V}\right)}(\varepsilon)$ (deformation with an infinite number of parameters). An analogue of the versal deformation $f$ is $f+F$, where $F\in \vartheta_{C^{8}\left(\overline{V}\right)}(\varepsilon )$, $\pi_{1}(F)\in B.$ Here $B$ is a versal subspace (\cite{Karpushkin1983}).

Let $f_{\pi}+P$ be a deformation, where $P\in B\subset E_{1}$.
\begin{definition}
We call $f_{\pi}+P$ a versal deformation of $f$, if for any deformation $F$, there exists  a vector $z(F)\in\mathbb{R}^{2}$ such that   the projection $\pi_{1}:(f+F)(z)\rightarrow P(F)$  is well defined, where $P(F)$ is the segment of the Taylor series (see \cite{Arn} page 37).
\end{definition}

Now we  consider the function $f+F$ and we represent it as
$$\begin{array}{l} {f+F=s_{10} x_{1} +s_{01} x_{2} +s_{20} x_{1}^{2} +2s_{11} x_{1} x_{2} +s_{02} x_{2}^{2} +s_{30}x_{1}^{3} +} \\ {+s_{21} x_{1}^{2} x_{2} +s_{12}x_{1} x_{2}^{2} +s_{03}x_{2}^{3}+s_{40} (x_{1} ,x_{2} )x_{1}^{4}+s_{31} (x_{1} ,x_{2} )x_{1}^{3}x_{2}+}\\{+s_{22} (x_{1} ,x_{2} )x_{1}^{2}x_{2}^{2}+s_{13} (x_{1} ,x_{2} )x_{1}x_{2}^{3}+s_{04}(x_{1} ,x_{2})x_{2}^{4},} \end{array}$$
where $s_{10} =\frac{\partial F(0,0)}{\partial x_{1} } ,\, \, s_{01} =\frac{\partial F(0,0)}{\partial x_{2} } ,\, \, s_{20} =\frac{1}{2} \frac{\partial ^{2} F(0,0)}{\partial x_{1}^{2} } ,\, \, s_{11} =\frac{1}{2} \frac{\partial ^{2} F(0,0)}{\partial x_{1} \partial x_{2} } ,\, \, s_{02} =\frac{1}{2} \frac{\partial ^{2} F(0,0)}{\partial x_{2}^{2} },$ $ s_{30} =\frac{1}{3} \frac{\partial ^{3} F(0,0)}{\partial x_{1}^{3}},\, \, s_{21} =\frac{1}{3} \frac{\partial ^{3} F(0,0)}{\partial x_{1}^{2}\partial x_{2}},\, \, s_{12} =\frac{1}{3} \frac{\partial ^{3} F(0,0)}{\partial x_{1}\partial x_{2}^{2}},\, \, s_{03} =\frac{1}{3} \frac{\partial ^{3} F(0,0)}{\partial x_{2}^{3}},$   $s_{k_{1}k_{2}}(x_{1},x_{2}):=s_{k_{1}k_{2}}=\int _{0}^{1}(1-u)^{2} \frac{\partial ^{4} (F+g)(ux_{1},ux_{2} )}{\partial x_{1}^{k_{1}} \partial x_{2}^{k_{2}}}du,$ $k_{1}+k_{2}=4.$

We change variables $x_{1} -z_{1}(F)=y_{1} $, $x_{2}-z_{2}(F)=y_{2},$ and expanding the function $(F+g)(y+z(F))$ by the Taylor formula at the point  $(y_{1},y_{2})=(0,0),$ we have
$$(f+F)(y+z(F))=\alpha_{00}(y,F)+\alpha_{10}(y,F)y_{1}+\alpha_{01}(y,F)y_{2}+$$$$+\alpha_{20}(y,F)y_{1}^{2}+\alpha_{11}(y,F)y_{1}y_{2}+
\alpha_{02}(y,F)y_{2}^{2}+\alpha_{30}(y,F)y_{1}^{3}+\alpha_{21}(y,F)y_{1}^{2}y_{2}+$$
$$+\alpha_{12}(y,F)y_{1}y_{2}^{2}+\alpha_{03}(y,F)y_{2}^{3}+\sum\limits_{i_{1}+i_{2}=4}\alpha_{i_{1}i_{2}}(y,F)y_{1}^{i_{1}}y_{2}^{i_{2}}+f_{\pi}(y_{1},y_{2}),$$
where
$\alpha_{00}(z)=f(0)+F(0)$ and $\alpha_{10}(y,F),\alpha_{01}(y,F),\alpha_{20}(y,F),\alpha_{11}(y,F),\alpha_{02}(y,F),$ $\alpha_{30}(y,F),$ $\alpha_{03}(y,F)$, $\alpha_{21}(y,F),\alpha_{12}(y,F)$ are functionals of $F$. Functionals $\alpha_{21}(y,F),$ $\alpha_{12}(y,F)$ are satisfying the conditions $\alpha_{21}(0,0)=0,\alpha_{12}(0,0)=0.$

\begin{proposition}\label{prop1}
There exists a positive number $\varepsilon>0$ such that for any $\left\|F\right\|_{C^{8}\left(\overline{V}\right)}<\varepsilon$ there exists a mapping $(z_{1},z_{2}):=\left(z_{1}(F),z_{2}(F)\right)\in C^{4}\left(U \to \mathbb R^{2} \right),$ defined in some neighborhood of $U$ for which the following equality holds
$$\pi _{1} \left(f\left(y_{1} +z_{1} ,y_{2} +z_{2} \right)+F\left(y_{1} +z_{1} ,y_{2} +z_{2} \right)\right)=\tilde{c}_{0} (F)+\tilde{c}_{1} (F)y_{1}+$$$$ +\tilde{c}_{2} (F)y_{2} +\tilde{c}_{3} (F)y_{_{1} }^{2}+\tilde{c}_{4} (F)y_{_{2}}^{2}+\tilde{c}_{5} (F)y_{_{1} }y_{2}+\tilde{c}_{6} (F)y_{_{1} }^{3}+\tilde{c}_{7} (F)y_{_{2} }^{3} ,$$
where \textbf{ }$\pi _{1} (\cdot )$ is the projection mapping from the space $C^{4} (V)$ to the space $E_{1}.$
\end{proposition}
Now consider the following functional equations for $\left(z_{1} ,z_{2} \right):$
\begin{equation} \label{GrindEQ__2_}
\Phi _{1} (y,F,z):=\alpha_{12}(z)=0,\Phi _{2} (y,F,z):=\alpha_{21}(z)=0.
\end{equation}

Let us give an auxiliary lemma.

\begin{lemma}\label{lem2} For continuous operators $\Phi _{1} (y,F,z)$, $\Phi _{2} (y,F,z)$ in the space $U_{1} \times C^{4} (V_{1} )\times U_{2} $ with $U_{1} \subset \mathbb R^{2} ,\, U_{2} \subset \mathbb R^{2},$ there exists a partial derivative with respect to $z$.
\end{lemma}
\begin{proof} For the sake of definiteness, let us show the existence of partial derivatives with respect to $z$ and differentiability of the operator $\Phi _{1} (y,F,z)$; for the function $\Phi _{2} (y,F,z)$ the proof is analogious.

Since $F\in C^{8} \left(\bar{V}\right)$  and  $g\in C^{8} \left(\bar{V}\right)$, this implies the differentiability of the operator $\Phi _{1} (y,F,z)$. The existence of derivatives of the mapping $\Phi _{2} (y,F,z)$ is considered similarly.
\end{proof}

\begin{lemma}\label{lem3} Operators $\Phi _{1} (y,F,z)$, $\Phi _{2} (y,F,z)$ satisfy $\Phi _{1}(y,0,0)\equiv0,$ $\Phi _{2}(y,0,0)\equiv0,$ and $\left|\begin{array}{l} {\frac{\partial \Phi _{1} }{\partial z_{1} } \, \, \, \, \, \frac{\partial \Phi _{1} }{\partial z_{2} } } \\ {\frac{\partial \Phi _{2} }{\partial z_{1} } \, \, \, \, \frac{\partial \Phi _{2} }{\partial z_{2} } } \end{array}\right|\ne 0.$
\end{lemma}
\begin{proof} From the explicit form of the operators $\Phi_{1},\Phi_{2}$ it follows that we have $\Phi _{1}(y,0,0)\equiv0$ and $\Phi _{2}(y,0,0)\equiv0.$ We also have, that $\left|\frac{\partial \Phi _{1} (0)}{\partial z_{1} } \frac{\partial \Phi _{2} (0)}{\partial z_{2} } -\frac{\partial \Phi _{1} (0)}{\partial z_{2} } \frac{\partial \Phi _{2} (0)}{\partial z_{1} } \right|=4\ne 0$ and hence the last statement is also true. Lemma \ref{lem3} is proved.

We proceed to the proof of Proposition \ref{prop1}. Since the operators $\Phi _{1} (y,F,z)$, $\Phi _{2} (y,F,z)$ by Lemma \ref{lem2} satisfy the conditions of the implicit mapping theorem, then there is a solution $z_{1} =z_{1} \left(y_{1} ,y_{2} ,F(y)\right)$, $z_{2} =z_{2} \left(y_{1} ,y_{2} ,F(y)\right)$ of the equation \eqref{GrindEQ__2_}, which are smooth functions depending on the smoothness of the mapping $F$.
\end{proof}
\bigskip
Let us assume that the function $\phi$ is a fixed
function of the form
\begin{equation}\label{funk1}\phi=x_{1}^{3}b(x)+x_{1}Q\left(x_{2},x_{3},...,x_{n}\right),
\end{equation}
where $b$ is a smooth function with $b(0)\neq0.$
Suppose in \eqref{funk1} that $Q\left(x_{2},x_{3},...,x_{n}\right)$ is an
analytic function satisfying the conditions $Q(0)=0,$ $\nabla
Q(0)=0.$  Suppose that $Q\not\equiv0$ and the pair
$\left(\beta(Q),p(Q)\right)$ is the type of asymptotics of the
oscillatory integral with phase $Q$ (see \cite{Karpushkin1983}).

\begin{lemma}\label{th5}
Let $Q$ be an analytic function satisfying the condition
$\beta(Q)\leq\frac{1}{4}.$ Then for the uniform type
oscillatory integral with phase $f$, the following relations hold:
$$\beta_{u}(\phi)=\beta(\phi)\,\,\,\,\,\ \text{and}\,\,\,\,\,\ p_{u}(\phi)=\left\{\begin{array}{l} {p(Q),\, \, \, \, \, \text{if}\, \, \beta(Q)<\frac{1}{4},} \\ {p(Q)+1,\, \, \, \, \text{if}\, \, \beta(Q)=\frac{1}{4}} \end{array}\right. .$$
\end{lemma}

We present the following lemma.
\begin{lemma}\label{th3} Let
\begin{equation}\label{integ2}J_{A}(\lambda,\sigma,a)=\int_{\mathbb{R}} e^{i\lambda \left(x^{3}b(x,\sigma)
+\sigma x\right)}a(x) dx,
\end{equation}
 where $b(0,0)\neq0,$ and  $a(x)$  has bounded variation and is supported in a sufficiently small neighborhood of the origin. Then there exists a constant $c$ such that for the  integral \eqref{integ2} the  following inequality holds uniformly:
$$\left|J_{A}(\lambda,\sigma,a)\right|\leq\frac{c\|a\|_{V}}{|\lambda|^{\frac{1}{3}}
+|\lambda|^{\frac{1}{2}}|\sigma|^{\frac{1}{4}}},$$ where
$\|a\|_{V}=|a(-\infty)|+V_{R}[a]$ and $V_{R}[a]$ is total variation, for all parameter $\sigma\in\mathbb{R}$.
\end{lemma}
For proof of Lemma \ref{th3} see \cite{IkromovMuller}, also \cite{Safarov}, \cite{Fed},\cite{popov1}.

{\emph{Proof of Lemma \ref{th5}}.} As we have the function
$\phi=x_{1}^{3}b(x)+x_{1}Q+\varphi,$ where $|\varphi|<\epsilon$, then we will change the variables $x_{1}=x_{1}\left(y_{1},x_{2},...,x_{n}\right)$
so that we have
$$\phi=y_{1}^{3}+\Sigma\left(x_{2},...,x_{n}\right)y_{1}+\psi\left(x_{2},...,x_{n}\right).$$
Moreover, $\Sigma$ has the form
$\Sigma=Q\Sigma_{1}\left(x_{2},...,x_{n}\right),$ where $\Sigma_{1}$ is an
analytic function satisfying the condition $\Sigma_{1}(0)\neq0.$
Then we have
$$\phi=y_{1}^{3}+Q\Sigma_{1}y_{1}+\psi,$$
and integral \eqref{int1} has the form
\begin{equation}\label{int3}
J(\lambda,f,a)=\int_{\mathbb{R}^{n}} e^{i\lambda
\phi}a(y_{1}) dy_{1}dx_{2}\dots dx_{n}.\end{equation} We consider the following cases:\\
Case 1. $\left|\lambda^{\frac{2}{3}}Q\Sigma_{1}\right|\leq M,$ where $M$ is sufficiently large and fixed.
Using Theorem 6.1 in the paper \cite{A} [page 438] and Lemma \ref{th3} we have:
$$|J|\leq\frac{\mu\left(\left|\lambda^{\frac{2}{3}}Q\Sigma_{1}\right|\leq M\right)}
{\lambda^{\frac{1}{3}}}\leq\frac{c\lambda^{-\frac{2}{3}\beta(Q\Sigma_{1})}
(\ln\lambda)^{p(Q)}}{\lambda^{\frac{1}{3}}}=\frac{c(\ln\lambda)^{p(Q)}}
{\lambda^{\frac{1}{3}+\frac{2}{3}\beta(Q)}}.$$
\\ Case 2.
$\left|\lambda^{\frac{2}{3}}Q\Sigma_{1}\right|> M.$ Consider the sets
$$A_{k}=\{2^{k}\leq\lambda^{\frac{2}{3}}|Q\Sigma_{1}|\leq2^{k+1}\}.$$
As $\Sigma_{1}(0)\neq0,$  without loss of generality we can assume that $\Sigma_{1}(0)=1$ at $|x|<\epsilon$ and
$\frac{1}{2}\leq\Sigma_{1}(x)\leq2.$ Hence
$$A_{k}\subset\left\{\frac{2^{k-1}}{\lambda^{\frac{2}{3}}}\leq|Q|\leq\frac{2^{k+2}}
{\lambda^{\frac{2}{3}}}\right\}.$$ For a measure of a set of smaller values we use Lemma $1^{'}$ in the paper \cite{Karpushkin1983} and we have: $$\mu\left(|Q|\leq\frac{2^{k+2}}{\lambda^{\frac{2}{3}}},
x\in U\right)
\leq\left(\frac{2^{k+2}}{\lambda^{\frac{2}{3}}}\right)^{\beta}
\left(\ln\left|\frac{\lambda^{\frac{2}{3}}}{2^{k+2}}\right|\right)^{p}.$$
Due to Lemma \ref{th3} for the integral
$$J_{k}=\int_{A_{k}}e^{i\lambda\left(y_{1}^{3}+y_{1}Q\Sigma_{1}\right)}
a(y_{1},x)dy_{1},$$ we find the following estimate:
\begin {eqnarray*}|J_{k}|=\left|\int_{A_{k}}e^{i\lambda\left(y_{1}^{3}+y_{1}Q\Sigma_{1}\right)}
a(y_{1},x)dy_{1}\right|\leq\frac{\mu(A_{k})}
{\lambda^{\frac{1}{2}}\left|Q\Sigma_{1}\right|^{\frac{1}{4}}}\leq
\left(\frac{2^{k+2}}{\lambda^{\frac{2}{3}}}\right)^{\beta}\\
\left(\ln\left|\frac{\lambda^{\frac{2}{3}}}{2^{k+2}}\right|\right)^{p}
\frac{1}
{\lambda^{\frac{1}{2}}}\left(\frac{\lambda^{\frac{2}{3}}}{2^{k+2}}\right)^{\frac{1}{4}}=
\frac{2^{(k+2)\beta-\frac{k-1}{4}}}{\lambda^{\frac{2}{3}\beta+\frac{1}{3}}}
\left(\ln\left|\frac{\lambda^{\frac{2}{3}}}{2^{k+2}}\right|\right)^{p}.
\end {eqnarray*}
From here we find the sum $J_{k}$ and, by estimating the integral $J$, we find the required
estimates. Lemma \ref{th5} is proved.
\color{black}

\section{A partition of unity}

The oscillatory integral is estimated using the partition of unity.

 Let $k=\left(k_{1}  ,k_{2} \right)$ and $r >0$ be fixed. We consider the mapping $\delta _{r }^{k} :\mathbb R^{2} \to \mathbb R^{2} $ defined by the formula:
$$\delta _{r} (x)=\left(r^{\frac{1}{4}} x_{1} ,r^{\frac{1}{4}} x_{2} \right).$$
Let us introduce a function $\beta (x)$, satisfying the conditions:\\
1) $\beta \in C^{\infty } \left(\mathbb R^{2} \right),$\\
2) $0\le \beta (x)\le 1$ for all $x\in \mathbb R^{2} $,\\
3) $\beta (x)=\left\{\begin{array}{l} {1,\, \, \, \, \, \text{when}\, \, |x|\le 1,} \\ {0,\, \, \, \, \text{when}\, \, \left|\delta _{2^{-1} } (x)\right|\ge 1.} \end{array}\right.$

The existence of such a function was proved in \cite{Vladim} (see also \cite{Sogge}).

\noindent Let us denote
$$\chi (x)=\beta (x)-\beta (\delta _{2} (x)).$$

\noindent The main properties of the function $\chi(x)$ are contained in the following lemma.
\begin{lemma}\label{lem4}
The  function $\chi(x)$ satisfies the following conditions:

\begin{enumerate}
\item  For any natural number $\nu_{0}\geq1$  and for an arbitrary fixed $x$, we have
$$\beta (\delta _{2^{-\nu_{0}}} (x))+\sum _{\nu =\nu_{0}}^{\infty }\chi \left(\delta _{2^{-\nu } } (x)\right)=1 .$$

\item For any $x\ne 0$ there exists $\nu _{0} (x)$  such that for any $\nu \notin [\nu _{0}(x),\nu _{0}(x)+4]$
$$\chi \left(\delta _{2^{\nu } } (x)\right)=0.$$

\item  For any $\nu _{0} $ there exists $\varepsilon >0$ such that $\chi \left(\delta _{2^{\nu } } (x)\right)=0$ for any $\nu <\nu _{0} $ and $\left|x\right|\ge \varepsilon $.
\end{enumerate}
\end{lemma}
\noindent Lemma \ref{lem4} was proved in the paper \cite{Sogge}.

\section{Proof of the main result}

Since the function has the form $f\left(x_{1} ,x_{2} \right)=f_{\pi}\left(x_{1} ,x_{2} \right) +g\left(x_{1} ,x_{2} \right)$, then applying Proposition \ref{prop1} to $f+F$ we get
\begin{equation} \label{GrindEQ__3_}
f+F=s_{00}+s_{10} y_{1} +s_{01} y_{2} +s_{20} y_{1}^{2}+ s_{02} y_{2}^{2} +s_{11} y_{1}y_{2} +s_{30}y_{1}^{3}+s_{03} y_{2}^{3}+f_{\pi}+R_{4} \left(y_{1} ,y_{2} \right),
\end{equation}
where $R_{4} \left(y_{1} ,y_{2} \right)$ is the remainder term. Now we estimate the integral $J.$ First, we introduce the "quasi-distance"  $\rho =\left|s_{10} \right|^{\frac{4}{3} } +\left|s_{01} \right|^{\frac{4}{3} } +\left|s_{20} \right|^{2}+\left|s_{02} \right|^{2}+\left|s_{11} \right|^{2}+\left|s_{30} \right|^{4}+\left|s_{03} \right|^{4}$.

First we consider the case $\lambda\rho\leq2$. Then we use change of variables as $y_{1}=\lambda^{-\frac{1}{4}}\tau_{1}$ and
$y_{2}=\lambda^{-\frac{1}{4}}\tau_{2}$. So we have
$$J\left(\lambda \right)=\lambda^{-\frac{1}{2}}\int_{\mathbb{R}^{2} }a\left(\lambda^{-\frac{1}{4} } \tau_{1},\lambda^{-\frac{1}{4} } \tau_{2} \right)e^{i\Phi_{1}} d\tau, $$
where $\Phi_{1}= \lambda^{\frac{3}{4} }s_{10} \tau _{1} +\lambda^{\frac{3}{4} }s_{01} \tau _{2} +\lambda^{\frac{1}{2} }s_{20} \tau _{1}^{2}+\lambda^{\frac{1}{2} }s_{02}  \tau _{2}^{2}+\lambda^{\frac{1}{2} }s_{11} \tau _{1}\tau_{2} +\lambda^{\frac{1}{4} }s_{30} \tau _{1}^{3}+\lambda^{\frac{1}{4} }s_{03} \tau _{2}^{3}+f_{\pi} + R_{4} \left(\lambda ^{-\frac{1}{4} } \tau _{1} ,\lambda ^{-\frac{1}{4} } \tau _{2} \right).$
We now apply Lemma \ref{lem4} for the integral $J(\lambda)$, to get the decomposition in the form
$$J\left(\lambda \right)=J_{0}(\lambda )+\sum _{k=k_{0}}^{\ln\lambda }J_{k}\left(\lambda \right),$$
where
$$J_{0}\left(\lambda\right)=\lambda^{-\frac{1}{2}}\int_{\mathbb R^{2}}a\left(\lambda^{-\frac{1}{4}}\tau_{1},\lambda^{-\frac{1}{4}}\tau_{2} \right)\beta_{0} (\delta _{2^{-k_{0}}} (x))e^{i\rho \Phi_1} d\tau,$$

$$J_{k}\left(\lambda\right)=\lambda^{-\frac{1}{2}}\int_{\mathbb R^{2}}a\left(\lambda^{-\frac{1}{4}}\tau_{1},\lambda^{-\frac{1}{4}}\tau_{2} \right)\chi\left(2^{-\frac{k}{4}}\tau_{1},2^{-\frac{k}{4}}\tau_{2}\right)e^{i\rho \Phi_{1}} d\tau.$$

First, we estimate the integral $J_{k} \left(\lambda \right)$. In this integral $J_{k} \left(\lambda \right)$  we make the change of variables  $2^{-\frac{k}{4} } \tau_{1} =t_{1}$, $2^{-\frac{k}{4}}\tau_{2}=t_{2}$, so that
$$J_{k} \left(\lambda \right)=2^{\frac{k}{2}}\lambda ^{-\frac{1}{2} }\int_{\mathbb R^{2}}a\left(2^{\frac{k}{4}}\lambda^{-\frac{1}{4} } t_{1},2^{\frac{k}{4}}\lambda^{-\frac{1}{4}}t_{2} \right)\chi \left(t_{1} ,t_{2} \right)e^{i2^{k}\Phi _{1k}\left(t,s \right)} dt ,$$
where the phase function has the form
$\Phi_{1k}= \frac{\lambda^{\frac{3}{4} }s_{10} }{2^{\frac{3k}{4}}} \tau _{1} +\frac{\lambda^{\frac{3}{4} }s_{01} }{2^{\frac{3k}{4}} } \tau _{2} +\frac{\lambda^{\frac{1}{2} }s_{20} }{2^{\frac{k}{2}} } \tau _{1}^{2}+\frac{\lambda^{\frac{1}{2} }s_{02} }{2^{\frac{k}{2}} } \tau _{2}^{2}+\frac{\lambda^{\frac{1}{2} }s_{11} }{2^{\frac{k}{2}} } \tau _{1}\tau_{2} +\frac{\lambda^{\frac{1}{4} }s_{30}}{2^{\frac{k}{4}} } \tau _{1}^{3}+\frac{\lambda^{\frac{1}{4} }s_{03}}{2^{\frac{k}{4}} } \tau _{2}^{3}+f_{\pi} +\frac{1}{2^{k}} R_{4} \left(2^{\, \frac{k}{4} }\lambda ^{-\frac{1}{4} } \tau _{1} ,2^{\, \frac{k}{4} }\lambda ^{-\frac{1}{4} } \tau _{2} \right).$

Since $\lambda\rho\leq2$ and $k$ is sufficiently large we may assume that coefficients are small. Then using  Van der Corput lemma \cite{VanDer} for the integral $J_{k} \left(\lambda \right),$ we have
\begin{equation}\label{F222}
|J_{k} \left(\lambda \right)|\leq\frac{2^{\frac{k}{2}}\lambda ^{-\frac{1}{2} }C}{2^{\frac{k}{2} }}=\frac{C}{\lambda^{\frac{1}{2} }}.
\end{equation}
Then,  summing the integrals $J_{k} \left(\lambda \right),$  we get
$$\sum_{k\leq\ln\lambda}J_{k}\left(\lambda \right)\leq\frac{C}{\lambda^{\frac{1}{2} }}\sum_{k\leq\ln\lambda}1\leq\frac{C\ln\lambda}{\lambda ^{\frac{1}{2} }}.$$

Now we consider the integral $J_{0} \left(\lambda \right)$. Since amplitude of the integral has compact support the trivial estimate yields
$$|J_{0} \left(\lambda \right)|\leq\frac{C}{\lambda ^{\frac{1}{2} }}.$$
In this case we get the result \eqref{Form1}.

Next we consider the case $\lambda\rho>2$. We change the variables $y_{1} =\rho ^{\frac{1}{4} } \tau _{1} ,\, y_{2} =\rho ^{\frac{1}{4} } \tau _{2} $. Then we have
$$J\left(\lambda \right)=\rho^{\frac{1}{2}}\int_{\mathbb{R}^{2}}a\left(\rho^{\frac{1}{4}}\tau_{1},\rho^{\frac{1}{4}}\tau_{2} \right)e^{i\lambda \rho\Phi} d\tau,$$
where $\Phi=\frac{s_{10}}{\rho^{\frac{3}{4}}} \tau_{1}+\frac{s_{01} }{\rho ^{\frac{3}{4} } } \tau _{2} +\frac{s_{20} }{\rho ^{\frac{1}{2} } } \tau _{1}^{2}+\frac{s_{02} }{\rho ^{\frac{1}{2} } } \tau _{2}^{2}+\frac{s_{11} }{\rho ^{\frac{1}{2} } } \tau _{1}\tau_{2} +\frac{s_{30}}{\rho ^{\frac{1}{4} } } \tau _{1}^{3}+\frac{s_{03}}{\rho ^{\frac{1}{4} } } \tau _{2}^{3}+f_{\pi} +\frac{1}{\rho } R_{4} \left(\rho ^{\frac{1}{4} } \tau _{1} ,\rho ^{\frac{1}{4} } \tau _{2} \right).$
We now apply Lemma \ref{lem4} for the integral $J(\lambda)$, to get the decomposition in the form
$$J\left(\lambda \right)=J_{0} (\lambda )+\sum _{k=k_{0} }^{\infty }J_{k} \left(\lambda \right),$$
where
$$J_{k} \left(\lambda \right)=\rho ^{\frac{1}{2} } \int _{\mathbb R^{2} }a\left(\rho ^{\frac{1}{4} } \tau _{1} ,\rho ^{\frac{1}{4} } \tau _{2} \right)\chi \left(2^{-\frac{k}{4} } \tau _{1} ,2^{-\frac{k}{4} } \tau _{2} \right)e^{i\lambda \rho \Phi} d\tau ,$$
$$J_{0} \left(\lambda \right)=\rho ^{\frac{1}{2} } \int _{\mathbb R^{2} }a\left(\rho ^{\frac{1}{4} } \tau _{1} ,\rho ^{\frac{1}{4} } \tau _{2} \right)\beta_{0} (\delta _{2^{-k_{0}}} (x))e^{i\lambda \rho \Phi} d\tau . $$

First, we estimate the integral $J_{k} \left(\lambda \right)$. In this integral $J_{k} \left(\lambda \right)$  we make the change of variables  $2^{-\frac{k}{4} } \tau _{1} =t_{1} $, $2^{-\frac{k}{4} } \tau _{2} =t_{2} $, so that
$$J_{k} \left(\lambda \right)=2^{\frac{k}{2} } \rho ^{\frac{1}{2} } \int _{\mathbb R^{2} }a\left(2^{\frac{k}{4} } \rho ^{\frac{1}{4} } t_{1} ,2^{\frac{k}{4} } \rho ^{\frac{1}{4} } t_{2} \right)\chi \left(t_{1} ,t_{2} \right)e^{i\lambda 2^{k} \rho \Phi _{k} \left(t,s,\rho \right)} dt ,$$
where the phase function has the form
$$\begin{array}{l} {\Phi _{k} \left(t,s,\rho \right)=2^{-\frac{3k}{4} } \sigma _{10} t_{1} +2^{-\frac{3k}{4} } \sigma _{01} t_{2} +2^{-\frac{k}{2} } \sigma _{20} t_{1}^{2} +2^{-\frac{k}{2} } \sigma _{02}t_{2}^{2} +2^{-\frac{k}{2} } \sigma _{11}t_{1}t_{2}} \\ {+2^{-\frac{k}{4} }\sigma_{30} t_{1}^{3} +2^{-\frac{k}{4} }\sigma_{03} t_{2}^{3}+f_{\pi}+2^{-k} \frac{1}{\rho } R_{4} \left(2^{\, \frac{k}{4} } \rho ^{\frac{1}{4} } t_{1} ,2^{\frac{k}{4} } \rho ^{\frac{1}{4} } t_{2} \right)\, ,} \end{array}$$
where $\sigma _{10} =\frac{s_{10} }{\rho ^{\frac{3}{4} } } $, $\sigma _{01} =\frac{s_{01} }{\rho ^{\frac{3}{4} } } $, $\sigma _{20} =\frac{s_{20} }{\rho ^{\frac{1}{2} } } $, $\sigma _{02} =\frac{s_{02} }{\rho ^{\frac{1}{2} } } $, $\sigma _{11} =\frac{s_{11} }{\rho ^{\frac{1}{2} } } $, $\sigma _{30} =\frac{s_{30} }{\rho ^{\frac{1}{4} } } $, $\sigma _{03} =\frac{s_{03} }{\rho ^{\frac{1}{4} } } $, $R_{4} \left(2^{\, \frac{k}{4} } \rho ^{\frac{1}{4} } t_{1} ,2^{\frac{k}{4} } \rho ^{\frac{1}{4} } t_{2} \right)=\frac{2^{k} \rho}{6} (s_{40} t_{1}^{4} +s_{31} t_{1}^{3} t_{2} +s_{22} t_{1}^{2} t_{2}^{2} +s_{13} t_{1} t_{2}^{3} +s_{04} t_{2}^{4} )$,

\noindent  where $s_{40} \left(t,s,\rho \right)=\frac{1}{6} \int _{0}^{1}(1-u)^{3} \frac{\partial ^{4} \Phi _{k} \left(ut,s,\rho \right)}{\partial t_{1}^{4} } du,\,  $ $s_{31} (t,s,\rho )=\frac{1}{6} \int _{0}^{1}(1-u)^{3} \frac{\partial ^{4} \Phi _{k} \left(ut,s,\rho \right)}{\partial t_{1}^{3} \partial t_{2}^{} } du,\,  $ $s_{22} \left(t,s,\rho \right)=\frac{1}{6} \int _{0}^{1}(1-u)^{3} \frac{\partial ^{4} \Phi _{k} \left(ut,s,\rho \right)}{\partial t_{1}^{2} \partial t_{2}^{2} } du,\,  $$s_{13} \left(t,s,\rho \right)=\frac{1}{6} \int _{0}^{1}(1-u)^{3} \frac{\partial ^{4} \Phi _{k} \left(ut,s,\rho \right)}{\partial t_{1}^{} \partial t_{2}^{3} } du,\,  $$s_{04} \left(t,s,\rho \right)=\frac{1}{6} \int _{0}^{1}(1-u)^{3} \frac{\partial ^{4} \Phi _{k} \left(ut,s,\rho \right)}{\partial t_{2}^{4} } du.\,  $

We can assume (depending on the support of the amplitude $\chi _{0} $), by Lemma \ref{lem4}, that the number $k_{0}$  is sufficiently large.

 We may assume $2^{k} \lambda\rho>L$ and let $k>k_{0}$ be a sufficiently large number. Then, by hypothesis, $\Phi _{k}$ can be considered a deformation of the function $f_{\pi}(\tau _{1},\tau _{2})$ and $(\tau_{1},\tau_{2})\in D:=$supp$(\chi)=\{\frac{1}{2}\leq|\tau|\leq2\}.$

If $\tau\neq0,$ using  Van der Corput lemma \cite{VanDer} (a more general statement is contained in \cite{Carberry},\cite{MichaelRuzhansky}), we have again estimate of the form \eqref{F222}.

Since support of the amplitude is on $2^{\frac{k}{4} } \rho ^{\frac{1}{4} } \leq C,$ then $2^{\frac{k}{4} } \lambda^{\frac{1}{4}}\rho ^{\frac{1}{4} } \leq C\lambda^{\frac{1}{4}}$ and $2^{\frac{k}{4} }\leq C\lambda^{\frac{1}{4}}$. So we have
$\sum_{k\leq C\ln\lambda}|J_{k}|\leq\frac{1}{\lambda^{\frac{1}{2}}}\sum_{k\leq C\ln\lambda}1\le\frac{C\ln\lambda}{\lambda^{\frac{1}{2}}} .$

Now  we consider estimates of the integral $J_{0}(\lambda )$. We consider several cases for parameters $\sigma$.

Let us introduce the quasisphere $\Omega:=\{\left|\sigma _{10} \right|^{\frac{4}{3} } +\left|\sigma _{01} \right|^{\frac{4}{3} } +\left|\sigma _{20} \right|^{2}+\left|\sigma_{02} \right|^{2}+\left|\sigma_{11} \right|^{2}+\left|\sigma_{30} \right|^{4}+\left|\sigma_{03} \right|^{4} =1 \}$ and consider the phase function
$$\begin{array}{l} {\Phi _{0} \left(\tau ,\sigma ,\rho \right)=\sigma _{10} \tau _{1} +\sigma _{01} \tau _{2} +\sigma _{20} \tau _{1}^{2}+\sigma _{02} \tau _{2}^{2}+\sigma _{11} \tau _{1}\tau_{2}  +\sigma_{30} \tau _{1}^{3} + } \\ {+\sigma_{03}\tau _{2}^{2}+s_{40} \tau _{1}^{4} +s_{31} \tau _{1}^{3} \tau _{2} +s_{22} \tau _{1}^{2} \tau _{2}^{2} +s_{13} \tau _{1} \tau _{2}^{3} +s_{04} \tau _{2}^{4} +f_{\pi}.} \end{array}$$

We note that on the quasi-sphere we have $c_{1}\leq\left|\sigma \right|\le c_{2}$, where $c_{1}, c_{2}$ are some fixed positive numbers.

Thus, the parameter space and $ $supp$\left(\beta(\delta_{2^{-k_{0}}}(\cdot)) \right)$ are compact sets. Let, $\sigma =\sigma ^{0} $, $\left|\sigma ^{0} \right|=c$ be fixed vector and let  $\tau =\tau ^{0}$ be a fixed point. Then $\Phi _{0} \left(\tau ,\sigma ,\rho \right)$ is a sufficiently small smooth deformation of the function

$$\begin{array}{l}
\Phi =\sigma _{10}^{0} \tau _{1} +\sigma _{01}^{0} \tau _{2} +\sigma _{20}^{0} \tau _{1}^{2}+\sigma _{02}^{0} \tau _{2}^{2}+\sigma _{11}^{0} \tau _{1}\tau_{2} +\sigma_{30}^{0} \tau _{1}^{3} +\sigma_{03}^{0}\tau _{2}^{2}+\\+s_{40} \tau _{1}^{4} +s_{31} \tau _{1}^{3} \tau _{2} +s_{22} \tau _{1}^{2} \tau _{2}^{2} +s_{13} \tau _{1} \tau _{2}^{3} +s_{04} \tau _{2}^{4} +f_{\pi}(\tau _{1},\tau _{2}) .\end{array}$$
If $\frac{\partial \Phi \left(\tau _{1}^{0} ,\tau _{2}^{0} \right)}{\partial \tau _{1} } \ne 0$ or $\frac{\partial \Phi \left(\tau _{1}^{0} ,\tau _{2}^{0} \right)}{\partial \tau _{2} } \ne 0$, then for $\left|\sigma -\sigma _{0} \right|<\varepsilon $ and  $\left|s_{40} \right|+\left|s_{31} \right|+\left|s_{13} \right|+\left|s_{04} \right|<\varepsilon $ the following estimate holds $\left|\nabla \Phi _{0} \left(\tau ,\sigma ,s\right)\right|>\delta >0$ for a positive number $\delta $.

Using the integration by parts formula for the integral
\begin{equation} \label{GrindEQ__4_}
J_{0}^{\chi } \left(\lambda \right)=\int _{\mathbb{R}^{2} }\chi \left(\tau \right)a\left(\tau _{1} ,\tau _{2} \right)\chi _{0} \left(\tau _{1} ,\tau _{2} \right)e^{i\lambda \rho \Phi _{0} \left(\tau ,\sigma ,s\right)} d\tau,
\end{equation}

we have
\begin{equation} \label{GrindEQ__3_}
\left|J_{0}^{\chi } \right|\le \frac{c\left\| a\right\| _{C^{1} } }{\left|\lambda \right|^{\frac{1}{2}}},
\end{equation}
with $\chi$ a smooth function concentrated in a sufficiently small neighborhood of the point $\tau ^{0}.$  It is enough to consider the case when $\tau^{0}=\left(\tau _{1}^{0} ,\tau _{2}^{0} \right)$ is a critical point.

 Since $\tau ^{0} $ is a critical point, we have the following equalities:
$$\sigma _{10}^{0} +2\sigma _{20}^{0} \tau _{1}^{0}+\sigma _{11}^{0} \tau _{2}^{0}+3\sigma _{30}^{0}(\tau _{1}^{0})^{2}+\partial_{1} f_{\pi}(\tau _{1}^{0},\tau _{2}^{0}) =0,$$$$\sigma _{01}^{0}+2\sigma _{02}^{0} \tau _{2}^{0}+\sigma _{11}^{0} \tau _{1}+3\sigma _{03}^{0} (\tau _{2}^{0})^{2}+\partial_{2} f_{\pi}(\tau _{1}^{0},\tau _{2}^{0}) =0.$$
The Hessian of the function $\Phi $ at the point $\left(\tau _{1}^{0} ,\tau _{2}^{0} \right)$  has the form
$$H=\left(\begin{array}{l} {2\sigma _{20}^{0}+6\sigma _{30}^{0}+\partial_{1}^{2}f_{\pi}\left(\tau^{0} \right) } \\ {\sigma _{11}^{0}+\partial_{1}\partial_{2}f_{\pi}\left(\tau^{0} \right) } \end{array}\right. \left. \begin{array}{l} {\sigma _{11}^{0}+\partial_{1}\partial_{2}f_{\pi}\left(\tau^{0} \right) } \\ {2\sigma _{02}^{0}+6\sigma_{03}^{0}\tau _{2}^{0}+\partial_{2}^{2}f_{\pi}\left(\tau^{0} \right) } \end{array}\right).$$

We consider two cases for the matrix $H$.
If this is a nonzero matrix, then we can use the generalised Van der Corpute lemma \cite{Carberry}, \cite{MichaelRuzhansky}, \cite{MichaelRuzhansky2012},
and we have a required bound.

Indeed, the rank of the Hessian $H$ is at least one. If the rank of matrix is one, then, using the Morse lemma with respect to parameters, for the integral $J_{0}$ we obtain the following estimate
$$\left|J_{0} \right|\le \frac{c\left\| a\right\| _{C^{1} }\rho^{\frac{1}{6}} }{\left|\lambda \right|^{\frac{1}{2} } }.$$

Assume now that the matrix $H$ is zero. Then $\sigma _{ij}^{0}=0,$ $i,j\leq2$, and $\sigma _{30}^{0}\neq0$ or $\sigma _{03}^{0}\neq0.$ If  $\sigma _{30}^{0}\neq0$ and $\sigma _{03}^{0}\neq0$ then we can use Duistermaat \cite{Duistermat} estimate and have a required bound. Suppose  $\sigma _{30}^{0}\neq0$ and $\sigma _{03}^{0}=0$ or  $\sigma _{30}^{0}=0$ and $\sigma _{03}^{0}\neq0.$  In this case we will use Lemma \ref{th5} and obtain
\begin{lemma}
We have
$$\left|\int _{U}e^{i\lambda \left(x_{1}^{3}b(x_{1},x_{2},\sigma)+x_{1}g_{1}(x_{2},\varepsilon)+g_{2}(x_{2},\varepsilon)\right)} a\left(x\right)dx \right|\le \frac{C\left\| a\right\| _{C^{1} }\ln(2+|\lambda|)}{\left|\lambda \right|^{\frac{1}{2} } },$$
where $g_{1}(x_{2},\varepsilon)$ is a sufficiently small perturbation of $b_{1}(x_{2},x_{2}^{4})$, where $b_{1}$ is a smooth function with $b_{1}(0)\neq0.$ Note that $\beta_{u}(g_{1}(x_{2},0))=\frac{1}{4}.$
\end{lemma}
\color{black}

Finally, summing up the estimates obtained, we arrive at the proof of Theorem \ref{theorem1}.

\textbf{Declaration of competing interest:}

This work does not have any conflicts of interest.

\section*{Acknowledgements} This paper was supported in parts by the FWO Odysseus 1 grant G.0H94.18N: Analysis and Partial Differential Equations and by the Methusalem programme of the Ghent University Special Research Fund (BOF) (Grant number 01M01021). The first author was supported by EPSRC grant EP/R003025/2. The second author was supported by "El-yurt umidi" Foundation of Uzbekistan and partially supported
in parts by the FWO Odysseus 1 grant G.0H94.18N: Analysis and Partial Differential Equations and by the Methusalem programme of the Ghent University Special Research Fund (BOF) (Grant number 01M01021).

\end{document}